\documentclass[11pt,a4paper]{article}
\usepackage[dvips]{graphicx}
\usepackage[utf8]{inputenc}
\usepackage{enumerate}
\usepackage[spanish,english]{babel}
\usepackage{amsthm,amssymb,amsfonts}
\usepackage[ruled,linesnumbered]{algorithm2e}
\usepackage{lineno}
\usepackage{hyperref}
\usepackage{caption}
\usepackage{subcaption}
\usepackage{xcolor}
\usepackage{cite}
\usepackage{pifont}
\usepackage{enumerate}
\usepackage{appendix}

\newtheorem{theorem}{Theorem}[section]

\newtheorem{proposition}{Proposition}[section]

\newtheorem{lemma}{Lemma}[section]
\newtheorem{corollary}{Corollary}[section]

\newtheorem{conjecture}{Conjecture}[section]

\theoremstyle{remark}
\newtheorem{remark}{Remark}[section]

\theoremstyle{remark}
\newtheorem{example}{Example}[section]

\newcommand\myparagraph[1]{\vspace{6pt}\noindent \textbf{#1}.}

\providecommand{\keywords}[1]{\textbf{\textit{Keywords:}} #1}

\title{New results on the robust  coloring problem}

\author{\normalfont\fontsize{12}{14}\selectfont
Delia Garijo$^1$\and
Alberto M\'arquez$^1$\and
Rafael Robles$^1$
}

\date{}
\begin{document}

\maketitle

\addtocounter{footnote}{1}  \footnotetext{Dep. Matem\'atica Aplicada I, Universidad de Sevilla, Spain. Emails: \{dgarijo, almar, rafarob\}@us.es}

\begin{abstract}
Many variations of the classical graph coloring model have been intensively studied due to their multiple applications; scheduling problems and aircraft assignments, for instance, motivate the \emph{robust coloring problem}. This model gets to capture natural constraints of those optimization problems by combining the information provided by two colorings: a vertex coloring of a graph and the induced edge coloring on a subgraph of its complement; the goal is to minimize, among all proper colorings of the graph for a fixed number of colors, the number of edges in the subgraph with the endpoints of the same color. The study of the robust coloring model has been focused on the search for heuristics due to its NP-hard character when using at least three colors, but little progress has been made in other directions. We present a new approach on the problem obtaining the first collection of non-heuristic results for general graphs; among them, we prove that robust coloring is the  model that better approaches the equitable partition of the vertex set, even when the graph does not admit a so-called \emph{equitable coloring}. We also show the NP-completeness of its decision problem for the unsolved case of two colors, obtain bounds on the associated robust coloring parameter, and  solve a conjecture on paths that illustrates the complexity of studying this coloring model.
\end{abstract}


\section{Introduction}

Coloring problems deal with partitioning the objects of a graph into  classes according to different criteria, and appear in many areas with seemingly no connection with coloring: time tabling and scheduling \cite{BMMPQ07}, frequency assignment \cite{SH97}, register allocation \cite{aplicacion2}, printed circuit board testing \cite{GJS76}, pattern matching \cite{O86} or analysis of biological and archeological data \cite{aplicacion1}; see also \cite{aplicacioningenieros} for descriptions of the first four mentioned applications.

 The classical coloring problem uses proper colorings: a \emph{$k$-proper coloring} is an assignment of $k$ colors to the vertices of a graph so that no edge has both endpoints of the same color. This is an NP-hard problem for three or more colors (and polynomial for two colors), which has received a large attention in the literature, not only for its real world applications, but also for its theoretical aspects and computational difficulty; see, for instance \cite{CZ20, GJ79}.

Other different criteria have been considered in coloring problems as the \emph{equitable} partition, that is, partitioning the vertex set of a graph into equal or almost equal subsets. Formally, a graph is \emph{equitable $k$-colorable} if it admits a $k$-proper coloring such that the cardinalities of any two color classes differ by at most one.

Equitable coloring of graphs was introduced by Meyer \cite{M73} for modeling problems in an operations research context, and has since been widely investigated due to its many practical applications in sequencing and scheduling; see for example \cite{FJK16} and the references therein, in particular, \cite{F06} for a specific application in scheduling. As explained in \cite{FJK16}, this type of coloring models situations in which one desires to split a system into equal or almost equal conflict-free subsystems. However, not every system admits such a division, and other criteria are needed in order to approach as much as possible the equitable partition; here arises the \emph{robust coloring problem}  (RCP, for short) that  can be stated as follows:

\begin{quote}
    Let $G=(V(G),E(G))$ be an unweighted and simple graph with chromatic number $\chi(G)$. Given a subgraph $H$ of its complement graph $\overline{G}$ and a positive integer $k\geq \chi(G)$, find a $k$-proper coloring $\phi$ of $G$ that minimizes, over all those $k$-proper colorings, the number of monochromatic edges\footnote{\emph{Monochromatic} edges are those whose endpoints have the same color; otherwise the edges are called \emph{bichromatic}.} in the induced coloring on $H$. We say that $\phi$ is a \emph{$k$-robust coloring of $(G,H)$}, and $m(G,H,k)$ is such minimum.
\end{quote}
The original statement of the problem introduced in \cite{YR03} considers $H$ to be a weighted graph, and the goal is to minimize the sum of the weights of the monochromatic edges. Our statement establishes all edge weights in $H$ to be 1, as both versions are equivalent for most of the questions addressed in this work, and for those that are not, our arguments can be adapted. We will be more precise on this issue in each of the sections below, once the different problems that we approach have been described in detail.

Applications of robust coloring are summarized in \cite{YR03, Lim05}; among them, it highlights applications to timetabling and scheduling problems, geographical maps, and aircraft assignment. For example, the aircraft assignment problem can be modeled as a graph coloring problem where each vertex in the graph $G$ represents a flight route and each color represents one aircraft. There is an edge between two vertices if an aircraft cannot serve the two flight routes represented by the two vertices. If flight-delays are taken into account, the overlap relationship between flight routes changes, and this information is captured by the edges of a new graph $H$. The monochromatic edges in $H$ represent cancelled flights, and the goal is to minimize the number of flights that must be cancelled when there are $k$ aircrafts.

\myparagraph{Related work}
As it was mentioned before, the RCP was introduced in \cite{YR03}, where the authors also describe several applications of this coloring model, and conclude that the decision problem is NP-complete for $k\ge 3$. They also present a binary programming model and outline a genetic algorithm. Due to their complexity result, most papers in the topic search for heuristics. In \cite{WX13} the authors develop several meta-heuristics to solve the RCP including genetic algorithm, simulated annealing and tabu search. A column generation-based heuristic algorithm is presented in \cite{YSH17}. A study on the robust aircraft assignment is developed in
\cite{Lim05}; the authors propose new techniques for an approximate solution of the problem, such as the partition based encoding and several meta-heuristics (local search, simulated annealing, tabu search and hybrid method).
Other references in this direction are \cite{ABH14, DPPP14}. Almost no progress has  been  made  in  other directions: we can only refer the reader to \cite{BRM05} for a theoretical study on the RCP for the case of $G$ and $H$ being paths on the same set of vertices and the value $k=3$. The authors present an exact but exponential algorithm to find a $3$-robust coloring of $(G,H)$, and a randomized algorithm and a greedy algorithm analyzing the cost of their output. They also apply their randomized algorithm to obtain bounds on the problem of maximizing the sum of the weights (costs) of the monochromatic edges of $H$, and pose two conjectures related with bounding the number of monochromatic edges of $H$ induced by a $3$-coloring of $G$.

\myparagraph{Our results} We present a non-heuristic approach to the RCP for general graphs. In
Section \ref{sec:equitable}, we first prove that robust coloring is the model that better approaches the equitable partition even when the graph has no equitable coloring; if the graph admits such a coloring, we extend the known connection between robust colorings of $(G, \overline{G})$ and equitable colorings of $G$ to a broad class of subgraphs $H$.
 The NP-complete nature of the decision problem of robust coloring is then established in Section \ref{sec:complexity} for the unsolved case of two colors, in contrast to the corresponding decision problems for classical graph coloring and equitable coloring. 
 Section \ref{sec:greedy} focuses on the  modifications required by the greedy algorithm for classical graph coloring in order to guarantee an optimal solution (for some vertex ordering) when dealing with robust colorings and equitable colorings. We introduce the \emph{robust-greedy algorithm} as the variation satisfying that property. This algorithm is  key in Section \ref{sec:bounds}, where we first obtain a tight upper bound on $m(G,H,k)$ for arbitrary graphs $G$ and subgraphs $H$, and then for graphs defined as \emph{$\alpha$-greedy orientable} (this includes trees, some series-parallel graphs and bipartite outerplanar graphs).
In Section \ref{subsec:paths}, we prove in the affirmative a conjecture posed by L\'opez-Bracho et al.~\cite{BRM05} on $m(G,H,3)$ for $G$ and $H$ being two paths on the same  set of vertices, and extend the result to $H$ being a vertex-disjoint union of paths;  the solution of this conjecture illustrates very well the complexity of dealing with robust colorings.

\vspace{0.15 cm}

We assume $k<{\rm min}\{|V(G)|, \chi(G)\chi(H)\}$ for otherwise $m(G,H,k)=0$ as there are always enough colors to obtain a proper coloring in $G$ and also in $H$. In addition, for short, we omit the term proper and simply say coloring when no confusion may arise.

\section{Robust colorings as an approach to equitable colorings}\label{sec:equitable}

When a graph $G$ has an equitable coloring, one obtains a uniform distribution of colors on the vertices, and the question is whether this decreases the number of monochromatic edges in a subgraph $H$ of $\overline G$, but this question can not be answered for an arbitrary $H$ as the answer would completely depend on its structure. Thus, it makes sense to study the connection between equitable colorings and robust colorings when $H$ is the whole  $\overline{G}$. In this section, we go further by setting $H$ as the induced subgraph in $\overline{G}$ by a subset of vertices $S\subseteq V(G)$. We denote this graph as $\overline{G}[S]$, see Figure  \ref{k33} for some examples.

For graphs $G$ that admit an equitable coloring, it is known that a $k$-coloring of $G$ is equitable if and only if it is a $k$-robust coloring of $(G, \overline{G})$; this can be deduced, for instance, from \cite[Proposition 3.1]{YR03}. We first prove that, even if $G$ does not admit an equitable coloring, robust coloring is the model that better approaches the uniform distribution of colors.

Let $S\subseteq V(G)$, and consider the color classes $C_1, \ldots ,C_k$ that partition $V(G)$ by a $k$-coloring $\phi$ of $G$ (each class contains the vertices with the corresponding color). Let $P_{\phi, S}=(n_1, \ldots, n_k)$ be the partition of $|S|$ associated to the coloring $\phi$, where $n_i=|C_i\cap S|$. We say that $\phi$  is \emph{equitable over} $S $ if $P_{\phi, S}$ satisfies that $|n_i-n_j| \leq 1$ for $1\leq i, j \leq k$; with some abuse of the language, we may indistinctly say that $S$ \emph{admits} an equitable $k$-coloring (see Figure \ref{k33}).
 There are monochromatic edges in the induced edge coloring of $\overline{G}[S]$ if and only if $n_i>1$ for some $1\leq i\leq k$; each pair of vertices of $S$ in the same color class $C_i$ determines a monochromatic edge, and so the total number of monochromatic edges induced by the partition $P_{\phi, S}$ in $\overline{G}[S]$, denoted by $m(P_{\phi, S})$, is
\begin{equation}\label{eq:mp}
m(P_{\phi, S}) = \sum_{1\leq i \leq k, n_i>1} {n_i \choose 2},\end{equation}
which can be rewritten as
$$m(P_{\phi, S})=\frac 12 \sum_{i=1}^{k} (n_i^2-n_i)=\frac 12 \sum_{i=1}^{k} ((n_i-\frac 12)^2-\frac 14)= -\frac k8+\frac 12 d^2(P_{\phi, S},Q),$$
where $d(P_{\phi, S}, Q)$ is the Euclidean distance between the point $P_{\phi, S}=(n_1, \ldots, n_k)$ and the point $Q=(\frac 12,\dots,\frac 12)$ in a $k$-dimensional space. (Note that, in the above equation, we include in the sum the case  $n_i=1$ since $n_i^2-n_i=0$.)

\begin{figure}[ht]
	\centering
\includegraphics[width=0.99\textwidth]{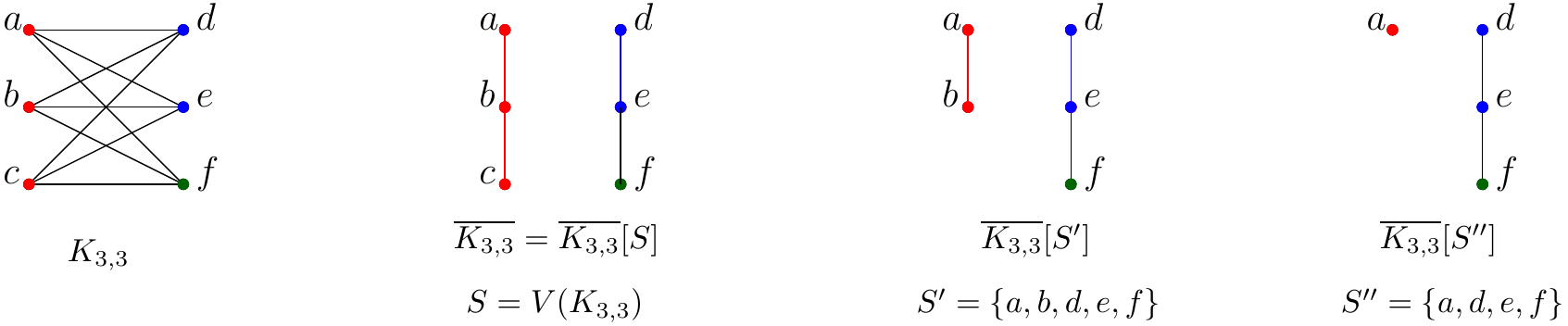}
	\caption{A $3$-coloring of $K_{3,3}$, and the induced edge colorings in $\overline{K_{3,3}}$,  $\overline{K_{3,3}}[S']$, and $\overline{K_{3,3}}[S'']$; the sets $S'$ and $S''$ admit an equitable $3$-coloring, and $S$ does not.}
 \label{k33}
\end{figure}

We can thus conclude that $m(P_{\phi, S})$ is mimimum over all $k$-colorings $\phi$ of $G$ if and only if $d(P_{\phi, S}, Q)$ is minimum.
Hence, minimizing the number of monochromatic edges in the induced coloring of $\overline{G}[S]$ is equivalent to finding the point $P_{\phi, S}$ in the hyperplane described by the equation $n_1+n_2+ \ldots +n_k=|S|$
that minimizes the distance to $Q$.
Further, $d(P_{\phi, S}, Q)$ is minimum if and only if $d(P_{\phi, S}, P_{|S|,k})$ is minimum, where $P_{|S|,k}=(\frac{|S|}{k}, \ldots, \frac{|S|}{k})$ is the orthogonal projection of $Q$ onto that hyperplane.
Observe that $P_{|S|,k}$ represents the ideal uniform distribution into the $k$ color classes. This distribution may not exist ($|S|$ might not even be divisible by $k$) but we have shown that the $k$-robust coloring is the closest to it under the Euclidean metric; this is the content of the following theorem.

\begin{theorem}\label{thm:ideal}
A $k$-coloring $\phi$ of a graph $G$ is a $k$-robust coloring of $(G, \overline{G}[S])$ if and only if $d(P_{\phi}, P_{|S|, k})$ is minimum over all $k$-colorings of $G$.
\end{theorem}

Consider now a $k$-robust coloring $\phi$ of $(G, \overline{G}[S])$ and the partition $P_{\phi, S}=(n_1,\dots, n_k)$. Suppose that $n_i<n_j$ for some $1\leq i,j\leq k$, and let
 $P=(n_1, \ldots, n_i+1, \ldots ,n_j-1, \ldots n_k)$; this is a $k$-partition of $|S|$ that is not necessarily associated to a $k$-coloring but, with some abuse of notation, we set $m(P)$ as
$$m(P)= {n_i+1 \choose 2} +{n_j-1 \choose 2}+\sum_{\ell \neq i,j} {n_{\ell} \choose 2}.$$
By equation (\ref{eq:mp}),  we have $m(P_{\phi, S})-m(P) = n_j-n_i-1\geq 0$, and so $m(P_{\phi, S})\geq m(P)$. Therefore, any partition of $|S|$ satisfying that any two of its elements differ by at most one  is a minimum of the function $m(\cdot)$ over all $k$-partitions of $|S|$. Thus, we have proved the following proposition, where we also give the minimum value of the function $m(\cdot)$, which is straightforward.

\begin{proposition}\label{thm:cotainf} 
For every $k\geq\chi(G)$ it holds that:
$$m(G,\overline{G}[S],k)\geq (k-r) {s \choose 2} + r {s+1 \choose 2},$$
where $s=\left\lfloor \frac{|S|}{k}\right\rfloor$ and $r=|S|-sk$. Moreover, the bound is tight if and only if $S$ admits an equitable $k$-coloring.
\end{proposition}

The preceding lower bound is the number of monochromatic edges in the induced edge coloring of $\overline{G}[S]$ by an equitable $k$-coloring over $S$, if it exists. Thus, we extend the known connection between equitable colorings and robust coloring to the graph $\overline{G}[S]$.

\begin{theorem}\label{col:rec}
Let $S\subseteq V(G)$ be a subset of vertices that admits an equitable $k$-coloring. Then,
a $k$-coloring $\phi$ of $G$ is equitable over $S$ if and only if $\phi$ is a $k$-robust coloring of $(G, \overline{G}[S])$.
\end{theorem}

We next illustrate the previous results with some examples.

\begin{example}
  Figure \ref{cometa-eps-converted-to} shows a $4$-equitable coloring of a graph $G$, which by  the relation between equitable colorings and robust colorings, is a $4$-robust coloring of $(G, \overline{G})$. The problem arises when a graph has no equitable coloring for some value $k$. This happens to
the complete bipartite graph $K_{3,3}$ when setting $k=3$ since any $3$-coloring generates color classes $C_1, C_2, C_3$ of cardinality $1,2,3$, respectively. See Figure  \ref{k33}. Theorem \ref{thm:ideal} establishes that the closest partition of the vertex set to an equitable partition is given by a $3$-robust coloring of $(G, \overline{G})$. Further, Theorem \ref{col:rec} allows us to study the scenario for subsets $S$ of vertices in $K_{3,3}$. All subsets $S$ containing at most two vertices of class $C_3$ admit an equitable $3$-coloring. Moreover, $m(G, \overline{G}[S],3)$ is either $1$ or $2$ (depending on the set $S$ considered).
\end{example}

\begin{example} The argument to prove Proposition \ref{thm:cotainf} can be used to obtain robust colorings of $(G, \overline{G}[S])$ when the graph $G$ does not admit equitable colorings. For instance, the wheel graph $W_{1,7}$ with $8$ vertices has no equitable colorings as there is always a color class of cardinality 1 (determined by the central vertex). For $k=4$, the above mentioned argument establishes that the partition $(1,1,3,3)$ can not be associated to a $4$-coloring of $W_{1,7}$ that is a $4$-robust coloring of $(W_{1,7}, \overline{W_{1,7}})$, but $(1,2,2,3)$ gives such a robust coloring.  \end{example}



\subsection{Complexity of robust colorings}\label{sec:complexity}

The classical graph coloring decision problem is NP-complete for $k\geq 3$ colors, but polynomial for $k=2$ \cite{GJ79}. The same happens for equitable $k$-coloring \cite{FJK16}. Now, consider the following problem:


\vspace{0.2cm}

\noindent {\sc \small{Robust-Coloring}}

\noindent {\sc \small{Instance:}}  A graph $G$, a subgraph $H$ of $\overline{G}$, a positive integer $k\leq |V(G)|$, and $\overline{m}\in \mathbb{N}$.

\noindent {\sc \small{Question:}}  Does a $k$-coloring of $G$ exist such that the number of monochromatic edges in the induced coloring on $H$ is at most $\overline{m}$?

\vspace{0.2cm}

\noindent A reduction to graph coloring shows the NP-completeness of {\sc \small{Robust-Coloring}} for $k\geq 3$ \cite[Proposition 3.2]{YR03}. We next prove that, surprisingly, this decision problem is NP-complete even for $k=2$.

\begin{theorem}\label{th:NProbust}
{\sc \small{Robust-Coloring}} is an NP-complete problem for $k=2$.
\end{theorem}

\begin{proof}
The problem is in NP since one can compute in polynomial time the number of monochromatic edges induced in $H$ by a given $k$-coloring of $G$, and check whether this number is at most $\overline{m}$. Consider now the following NP-complete problem \cite{GJS76(1)}:

\vspace{0.2cm}

\noindent {\sc \small{Simple-Max-Cut}}

\noindent {\sc \small{Instance:}}  A graph $G=(V,E)$, $\ell\in \mathbb{N}$.

\noindent {\sc \small{Question:}} Does there exist a set $S\subset V$ such that $|\{su\in E \, | \, s\in S, u\in V-S\}| \geq \ell$?

\vspace{0.2cm}

\noindent We next reduce {\sc \small{Simple-Max-Cut}} to our decision problem, thus proving the result. Let $G=(V,E)$ be a graph with $n$ vertices, and let $\ell \in \mathbb{N}$. Let $V_n$ be the trivial graph with $n$ vertices (i.e., it has no edges); the graph $G$ is a subgraph of $\overline{V_n}$. Any $2$-coloring of $V_n$ induces a partition of $V$ into two subsets $S$ and $V-S$ such that the bichromatic edges in the induced coloring on $G$ are precisely the set $\{su\in E \, | \, s\in S, \, u\in V-S\}$. Therefore,
$$|\{su\in E \, | \, s\in S, u\in V-S\}| \geq \ell \Longleftrightarrow m(V_n, G, 2)\leq |E|-\ell.$$
The inequation $m(V_n, G, 2)\leq |E|-\ell$ is equivalent to the existence of a $2$-coloring of $V_n$ such that the number of monochromatic edges in the induced coloring on $G$ is at most $|E|-\ell$.
\end{proof}

\subsection{Robust-greedy algorithm}\label{sec:greedy}

For classical graph coloring, it is well-known that there always exists a vertex ordering in any graph such that the greedy algorithm\footnote{Recall that the greedy algorithm for graph coloring considers an ordering of the vertices of the graph and assigns to each vertex its first available color (i.e., the first color that has not been assigned to any of its already colored neighbours).} gives an optimal proper coloring, that is, a proper coloring using the minimum number of colors. However, this is not true for robust coloring, and neither for equitable coloring as the example in Figure~\ref{cometa-eps-converted-to} shows. We next introduce a variation of the greedy algorithm, called the \emph{robust-greedy algorithm}, which captures the constraints of the robust colorings and gives, for some ordering of the vertices of any graph, the closest partition to the equitable partition. Further,
this algorithm will lead, together with the notion of \emph{$\alpha$-greedy orientable graph} (introduced in Section \ref{sec:greedy orientable}), to upper bounds on $m(G,H,k)$ for well-known families of graphs $G$ and arbitrary subgraphs $H$ of $\overline{G}$.

\begin{figure}[ht]
	\centering
\includegraphics[width=0.3\textwidth]{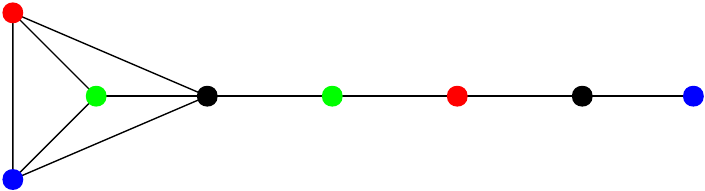}
	\caption{A graph $G$ that admits $4$-equitable colorings (as the one shown), which are $4$-robust colorings of $(G,\overline G)$. None of them can be obtained by the greedy algorithm with any of the $8!$ possible vertex orderings.}
 \label{cometa-eps-converted-to}
\end{figure}

As the classical greedy algorithm for graph coloring, the robust-greedy algorithm also processes the vertices of a graph $G$ in a given ordering, and there is an ordered list of colors (they are simply taken in order). In addition, we must keep track of  the number of monochromatic edges on a fixed subgraph $H$ of $\overline{G}$.
\begin{quote}
\noindent {\sc \small{Robust-greedy algorithm}}

Each vertex $v$ of $G$ is given the first color $c$ of the  list satisfying the two following  properties:

\begin{itemize}
    \item[(a)] color $c$ is available for $v$, i.e., it has not been assigned to the already colored neighbours of $v$ in $G$;
    \item[(b)] it minimizes, among all available colors for $v$, the number of monochromatic edges on $H$ with $v$ as an endpoint.
\end{itemize}
The algorithm stops when all vertices of $G$ have been colored.
\end{quote}

As we pointed out before, some questions on robust coloring cannot be approached for general subgraphs $H$ of $\overline G$ since the answer would depend on the structure of $H$, and it makes then sense to set $H=\overline{G}$. This happens in the following theorem.

\begin{theorem}\label{thm:ordering}
Let $G$ be a graph, and let $k\geq \chi(G)$ be a positive integer. There always exists a vertex ordering of $G$ such that the robust-greedy algorithm provides a $k$-robust coloring of $(G, \overline{G})$.
\end{theorem}

\begin{proof}
Let $\phi$ be a $k$-robust coloring  of $(G, \overline{G})$, and consider the color classes $C_i=\{u_1^i, \ldots, u_{n_i}^i\}$, $1\leq i\leq k$, in which $\phi$ partitions $V(G)$. Assume that the classes are ordered by increasing cardinality: $n_i\leq n_j$ for $1\leq i<j\leq k$.

Let $\mathcal{O}$ be a vertex ordering obtained by choosing a vertex from each class $C_i$ in a cyclic way (in increasing order) until there are no vertices left in any of the classes, for example, $\mathcal{O}$ could be: $u_1^1, u_1^2, \ldots, u_1^k, u_2^1, u_2^2, \ldots, u_2^k, \ldots, u_{n_1}^1, \ldots, u_{n_1}^k, u_{n_1+1}^2, \ldots, u_{n_1+1}^k, \ldots, u_{n_k}^k$.

The robust-greedy algorithm assigns color 1 to $u_1^1$ (the same first color as $\phi$), and when it processes $u_1^2$ it may happen that: (i) $u_1^1u_1^2 \in E(G)$ and so $u_1^2$ would be assigned color 2 by condition (a) of the algorithm, or (ii) $u_1^1u_1^2 \in E(\overline{G})$ and so, by condition (b) of the algorithm, $u_1^2$ would also be assigned color 2 (this color minimizes, among colors 1 and 2, the number of monochromatic edges in $\overline{G}$ with $u_1^2$ as endpoint). Hence, the robust-greedy algorithm assigns the same colors as $\phi$ to $u_1^1$ and $u_1^2$. This argument can be extended to the first $kn_1$ vertices, to which the algorithm assigns the same colors as $\phi$ generating $k$ color classes with the same size $n_1$ (after processing the vertices $u_1^1, u_1^2, \ldots, u_1^k, \ldots, u_{n_1}^1, \ldots, u_{n_1}^k$ of the ordering $\mathcal{O}$).

Now, the algorithm could assign the same colors as $\phi$ to the remaining vertices but, if at some later stage, the robust-greedy algorithm assigns to a vertex a different color than that assigned by $\phi$, we stop the algorithm and color the remaining vertices with the same colors as $\phi$, obtaining a new $k$-coloring $\psi$. The associated partitions $P_{\phi,V(G)}$ and $P_{\psi, V(G)}$ only differ in one element: roughly speaking, one vertex has changed from a bigger color class to a smaller one. Following the same argument as for Proposition \ref{thm:cotainf}, we obtain that $m(P_{{\phi},V(G)})\geq m(P_{{\psi},V(G)})$. Since $\phi$ is a robust coloring then $m(P_{{\phi},V(G)})= m(P_{{\psi},V(G)})$. For each change of color produced by the robust-greedy algorithm, we can argue as above obtaining a sequence of $k$-robust colorings of $(G, \overline{G})$ that lead to the desired $k$-robust coloring generated by the algorithm.
\end{proof}

The analogous of Theorem \ref{thm:ordering} for equitable partitions of vertex sets is obtained from Theorem \ref{col:rec} by setting $S=V(G)$.

\begin{corollary}
For every equitable $k$-colorable graph $G$, there always exists a vertex ordering such that the robust-greedy algorithm provides an equitable $k$-coloring of its vertices.
\end{corollary}

\section{Upper bounds on $m(G,H,k)$ for arbitrary $H$}\label{sec:bounds}

In this section we deal with arbitrary subgraphs $H$ of $\overline{G}$\footnote{Our results consider $H$ to be unweighted but our arguments can be easily adapted for multigraphs and graphs with rational edge weights; in the case of real edge weights, we can approximate them (using rational weights) with the desired precision.}. We first present a tight upper bound on $m(G,H,k)$ for every graph $G$, for which we need the following technical lemma, where two distinct colorings are considered, one of them not necessarily proper. Thus, to avoid any confusion, the term proper will not be omitted in Lemma \ref{lem:cota} and Theorem \ref{thm:ub}.

\begin{lemma}\label{lem:cota}
Let $t\geq 2$. For every proper $t$-coloring of a graph $G$ and every positive integer $t'\in [1,t]$ there exists a $t'$-coloring of $G$  that induces at most $|E(G)|\cdot\frac{2(t-t')}{tt'}$ monochromatic edges in $G$.
\end{lemma}

\begin{proof}
The result is straightforward for $t'=1$ as $|E(G)|$ is the number of monochromatic edges induced by any $1$-coloring of $G$ and $\frac{2(t-1)}{t}\geq 1$ for $t\geq 2$. If $t'=t$ the result establishes that there are no induced monochromatic edges, which is true for any proper $t$-coloring of $G$.

Assume now that $1<t'<t$, and let $\phi$ be a proper $t$-coloring of $G$, which induces an edge coloring of $G$ (according to the colors of the endpoints of the edges) with  ${t \choose 2}$ edge color classes. On average, each of these classes contains  $\mu_0 = |E(G)|/{t\choose 2}$ bichromatic edges. Consider a  color class with smallest cardinality, say that it corresponds to color $ij$. We obtain a $(t-1)$-coloring $\phi'$ from $\phi$ by identifying colors $i$ and $j$. Observe that the number of monochromatic edges induced by $\phi'$ in $G$ is at most $\mu_0$. The same argument applies to the coloring $\phi'$ and the value $\mu_1 = |E(G)|/{t-1\choose 2}$. Thus, after a reduction of $t-t'$ colors, we obtain a $t'$-coloring of $G$ and the corresponding values $\mu_0, \mu_1, \ldots ,\mu_{t-t'+1}$; this coloring induces at most the desired number of monochromatic edges:
$$\mu_0+\mu_1+\dots+\mu_{t-t'+1}=|E(G)|\cdot\left[ \frac{1}{{t\choose 2}}+\frac{1}{{t-1\choose 2}}+\dots+\frac{1}{{t-t'+1\choose 2}}\right]=|E(G)|\cdot\frac{2(t-t')}{tt'}$$
\end{proof}

With Lemma \ref{lem:cota} in hand, we obtain an upper bound on $m(G,H,k)$ for  general graphs $G$ and  arbitrary subgraphs $H$ of $\overline G$, in terms of several parameters: the chromatic numbers of $G$ and $H$, the number of edges of $H$, and the value $k$.

\begin{theorem}\label{thm:ub}
Let $G$ be a graph with $\chi(G)=p$, and let $H$ be a subgraph of $\overline G$ with $\chi(H)=q$. For every $k\geq {\rm max}\{p,q\}$ it holds that
$$m(G,H,k)\leq \frac{2|E(H)|}{pq}\left[\frac{(p-r)(q-s)}{s}+\frac{r(q-s-1)}{(s+1)}\right],$$
where $s = \left\lfloor \frac{k}{p}\right \rfloor$ and  $r=k-ps$. Moreover, the bound is tight.
\end{theorem}

\begin{proof}
Let $\phi$ be a $p$-proper coloring of $G$ with color set $\{1,2,\ldots,p\}$ and color classes $C_1,\ldots, C_p$ that partition $V(G)$. The induced coloring on $H$ partitions $E(H)$ into the sets $E_b$ and $E_{m,i}$ that contain, respectively, the bichromatic edges and the edges with assigned color $i$ in both endpoints for $1\leq i \leq p$. Assume that the sets $E_{m,i}$ are ordered by increasing cardinality.

Let $H_i$ be the graph with vertex set $C_i$ and edge set $E_{m,i}$; this graph might have a number of isolated vertices as there could be vertices of $G$ in $C_i$ that are endpoints of no monochromatic edge in $E_{m,i}$. By construction, $E(H_i)\subset E(H)$ and so $\chi(H_i)\leq \chi(H)=q$. Thus, there exists a $q$-proper coloring of $H_i$ and, by Lemma \ref{lem:cota}, for every $1\leq i \leq p-r$ there is an $s$-coloring $\phi_i$ of $H_i$ that induces at most $|E_{m,i}|\cdot\frac{2(q-s)}{qs}$ monochromatic edges in $H_i$. Note that Lemma \ref{lem:cota} can be applied by setting $t=q$ and $t'=s$ since it is assumed that $k<{\rm min}\{|V(G)|, pq\}$, which implies $s = \left\lfloor \frac{k}{p}\right \rfloor <q$.
Analogously, for every $p-r+1\leq i\leq p$, we obtain at most $|E_{m,i}|\cdot\frac{2(q-s-1)}{q(s+1)}$ monochromatic edges in $H_i$ induced by an $(s+1)$-coloring $\psi_i$ (Lemma \ref{lem:cota} is here applied for $t=q$ and $t'=s+1$; note that $s+1\leq q$.)

The union of all $s$-colorings $\phi_i$ and all $(s+1)$-colorings $\psi_i$  is a $k$-proper coloring of $G$ since each coloring is acting on a class $C_i$ whose vertices can have the same color. Further, this union of colorings, when restricted to $V(H)$, is a $k$-coloring of $H$ with at most
$$\frac{2(q-s)}{qs}\sum_{i=1}^{p-r} |E_{m,i}|+\frac{2(q-s-1)}{q(s+1)}\sum_{i=p-r+1}^{p} |E_{m,i}|$$ monochromatic edges. This value can be rewritten as
\begin{equation}\label{eq:mon}\frac{2(q-s-1)}{q(s+1)}\sum_{i=1}^{p} |E_{m,i}|+\left[\frac{2(q-s)}{qs}-\frac{2(q-s-1)}{q(s+1)}\right]\sum_{i=1}^{p-r} |E_{m,i}|.\end{equation}
Let $\mu=\frac{1}{p}(\sum_{i=1}^{p} |E_{m,i}|)$. As $\frac{2(q-s)}{qs}>\frac{2(q-s-1)}{q(s+1)}$ and $|E_{m,i}|\leq |E_{m,j}|$ for $i\leq j$, expression (\ref{eq:mon}) is at most
$$ \frac{2(q-s-1)}{q(s+1)}p\mu+\left[\frac{2(q-s)}{qs}-\frac{2(q-s-1)}{q(s+1)}\right](p-r)\mu,$$
which is
$$\frac{2\sum_{i=1}^{p}|E_{m,i}|}{pq}\left[\frac{(p-r)(q-s)}{s}+\frac{r(q-s-1)}{s+1}\right].$$
The desired bound is thus obtained as $\sum_{i=1}^{p} |E_{m,i}|\leq |E(H)|$.
Further, it is attained by the graph $G$ and the subgraph $H$ of $\overline{G}$ in Figure \ref{cota}, and the value $k=3$; the subgraph $H$ is isomorphic to $\overline{G}[S]$ for $S=\{a,b,c\}$. We have $\chi(G)=p=3$ and $\chi(H)=q=2$. For $k=3$, we obtain $s=0$ and $r=1$. The upper bound is then equal to $1$, and $m(G,H,3)=1$ (this follows from the fact that $G$ has a unique 3-proper coloring, also shown in Figure \ref{cota}.)
\end{proof}

\begin{figure}[ht]
	\centering
\includegraphics[width=0.7\textwidth]{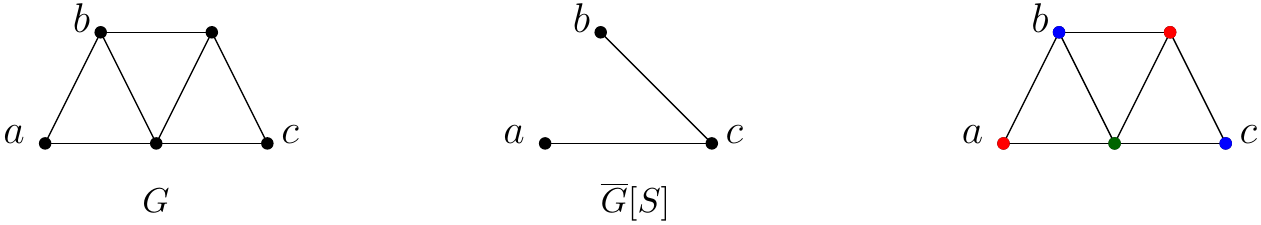}
	\caption{An example of graph $G$ that attains the bound of Theorem \ref{thm:ub} together with the graph $\overline{G}[S]$ for $S=\{a,b,c\}$ and the value $k=3$. On the right it is shown the unique $3$-proper coloring of $G$, which generates only one monochromatic edge in $\overline{G}[S]$.}
 \label{cota}
\end{figure}

\begin{remark}\label{rem:s1}
For $s = \left\lfloor \frac{k}{p}\right \rfloor=1$, which is the most frequent case (in general the value $k$ does not double the chromatic number of the graph), the upper bound of Theorem \ref{thm:ub} is
$$m(G,H,k)\leq 
\frac {| E(H) |} {p} \left (p-  \frac {2r} {q} \right).$$
\end{remark}


\subsection{$\alpha$-greedy orientable graphs}\label{sec:greedy orientable}

We now focus on graphs that satisfy a property called $\alpha$-\emph{greedy orientable}; among this type of graphs highlight: trees, simple series-parallel graphs with chromatic number 3, and simple bipartite outerplanar graphs.
To define this key property, we first fix a vertex ordering in a graph $G$ and consider the induced orientation on its edges, that is, vertex $u$ is the tail and $v$ is the head of the oriented edge $(u,v)$ if and only if $u<v$ in the ordering. Let $\overrightarrow{G}$ be the resulting oriented graph, and let $\Delta^{-}(\overrightarrow{G})$ be its maximum in-degree.
Observe that the greedy algorithm applied to our vertex ordering gives a coloring of $G$ with at most $\Delta^{-}(\overrightarrow{G})+1$ colors, and so
$\chi(G)\leq  \Delta^{-}(\overrightarrow{G})+1$. We say that $G$ is \emph{$\alpha$-greedy orientable} for $\alpha \geq 0$ if there is an ordering of its vertices such that $\Delta^{-}(\overrightarrow{G})=\chi(G)-1+\alpha.$ Figure~\ref{alpha_greedy} illustrates an example.

Our interest is to find vertex orderings whose associated $\alpha$ is the smallest possible value; this will give interesting upper bounds on $m(G,H,k)$ as we show next.

\begin{figure}[ht]
	\centering
\includegraphics[width=0.7\textwidth]{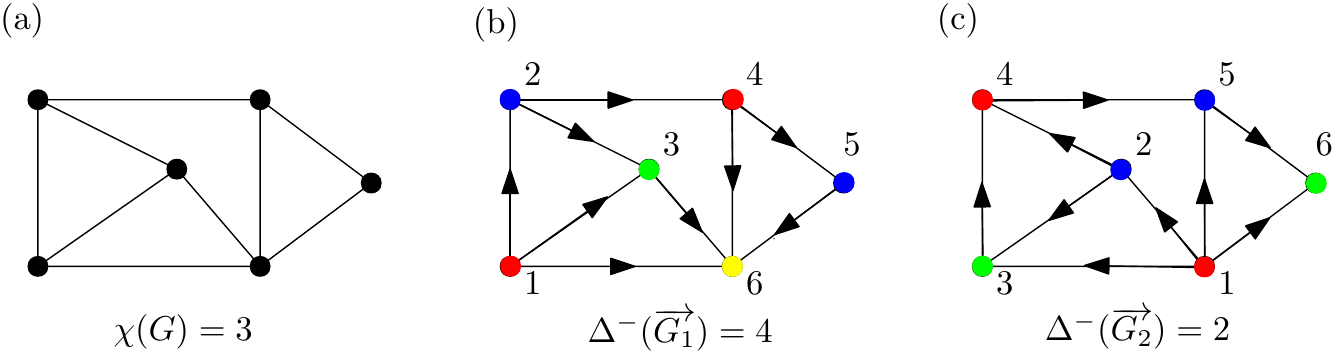}
	\caption{For the graph $G$ in (a) we consider two vertex orderings in (b) and (c), and their respective greedy colorings. The ordering in (b) shows that $G$ is $1$-greedy orientable, and $G$ is $0$-greedy orientable with the ordering in (c).} \label{alpha_greedy}
\end{figure}

Consider an $\alpha$-greedy orientable graph $G$, and let $k\geq \alpha+\chi(G)$. When we apply the robust-greedy algorithm (see Section \ref{sec:greedy}) to the vertex ordering in $G$ associated to $\alpha$, there are always $k-\Delta^{-}(\overrightarrow{G})$ available colors when a vertex $v$ is processed by the algorithm, that is, $k-\chi(G)+1-\alpha$ available colors. Further, $v$ is assigned a color that minimizes, at that stage, the number of monochromatic edges in $H$ with $v$ as endpoint. This implies that at most one $(k-\chi(G)+1-\alpha)$--th of the edges with $v$ as an endpoint are monochromatic. This proves the following upper bound for arbitrary subgraphs $H$ of $\overline{G}$.

\begin{theorem}\label{th:alphagreedy} Let $G$ be an $\alpha$-greedy orientable graph, and let $k\geq \alpha+\chi(G)$. Then,
	$$m(G,H,k)\leq \frac{|E(H)|}{k-\chi(G)+1-\alpha}$$ for every subgraph $H$ of $\overline{G}$.
\end{theorem}

The preceding theorem leads to upper bounds on $m(G,H,k)$ for well-known families of graphs; to apply it, we must obtain values of $\alpha$ for which these graphs are greedy orientable.

\begin{proposition}\label{pro:greedyfamilies} The following statements hold.
	\begin{itemize}
		\item[(i)] Trees are the unique bipartite graphs that are $0$-greedy orientable.
\item[(ii)] Simple series-parallel graphs with chromatic number 3 are $0$-greedy orientable.
\item[(iii)] Simple bipartite outerplanar graphs are $1$-greedy orientable.
\item[(iv)] Every $\ell$-tree\footnote{An $\ell$-tree is a graph formed by starting with a  complete graph on $(\ell + 1)$ vertices and then repeatedly adding vertices in such a way that each added vertex has exactly $\ell$ neighbors that, together, the $\ell + 1$ vertices form a clique. Note that $1$-trees are the same as unrooted trees, and $2$-trees are maximal series-parallel graphs that also include the maximal outerplanar graphs.} with chromatic number $\ell+1$ is $0$-greedy orientable.
\end{itemize}
\end{proposition}

\begin{proof}
	
(i) Let $T$ be a tree, viewed as a rooted tree, and  consider all edges oriented away from the root. This gives a vertex ordering satisfying that $\Delta^{-}(\overrightarrow{T})=1$ and so $T$ is $0$-greedy orientable. Now, if a bipartite graph $G$ is $0$-greedy orientable then $\Delta^{-}(\overrightarrow{G})=1$, which implies that $G$ cannot contain a cycle as any vertex ordering in the cycle would force at least one vertex to have in-degree bigger than $1$. Thus, $G$ is a tree.	

(ii)  By definition, any simple series-parallel graph $G$ can be turned into $K_2$ by a sequence of two operations: identifying any vertex $v$ of degree two with one of its adjacent vertices (so the two incident edges with $v$ are replaced by one), and deleting parallel edges (leaving just one copy). This series-parallel reduction produces a vertex ordering in $G$ such that $\Delta^{-}(\overrightarrow{G})=2$. Indeed, it suffices to label the vertices of $G$ following the sequence in the reverse order: the vertices of $K_2$ are assigned $1$ and $2$, and every time a new vertex appears in the sequence is given the following number; see Figure \ref{seriesparallel}. Since $\chi(G)=3$, then $\alpha=\Delta^{-}(\overrightarrow{G})-\chi(G)+1=0$.

(iii) The result follows from the fact that any outerplanar graph $G$ is series-parallel. Thus, we can take the same vertex ordering as above, which gives $\Delta^{-}(\overrightarrow{G})=2$. Since $G$ is also bipartite, we have $\chi(G)=2$ and $\alpha=\Delta^{-}(\overrightarrow{G})-\chi(G)+1=1$.

(iv) The inductive construction of an $\ell$-tree $G$ when $\ell=\chi(G)-1$ provides a vertex ordering in $G$ satisfying that  $\Delta^{-}(G)=\chi(G)-1$ (the vertices of the original complete graph are labelled and every vertex that is added has in-degree $\ell$, equal to the number of neighbors at that stage of the construction). Thus, $\alpha=\Delta^{-}(\overrightarrow{G})-\chi(G)+1=0$.
\end{proof}

\begin{figure}[ht]
	\centering
\includegraphics[width=0.99\textwidth]{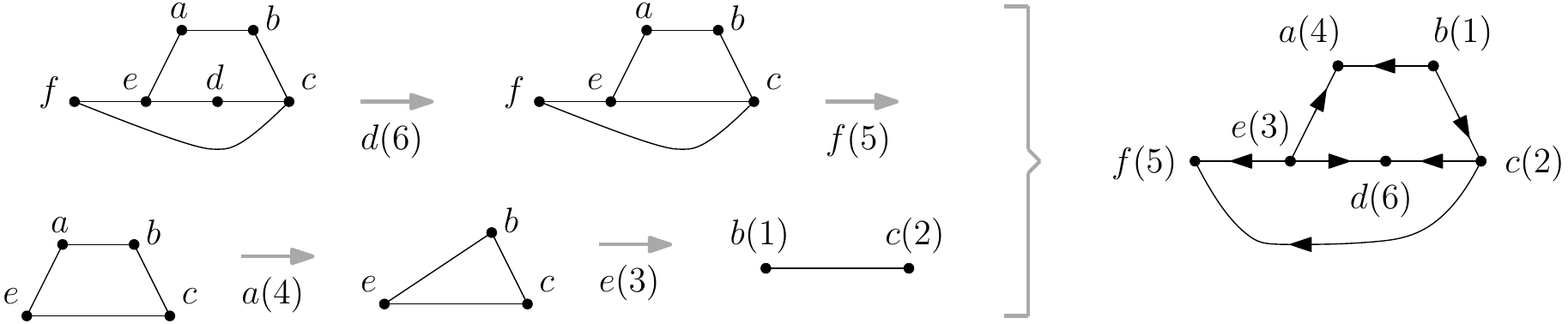}
\caption{The reverse order of the series-parallel reduction generates a vertex ordering (numbers in brackets); the associated oriented graph is shown on the right. For short, the deletion of the parallel edges $ec$ and $bc$ is not indicated. } \label{seriesparallel}
\end{figure}

 Theorem \ref{th:alphagreedy} and Proposition \ref{pro:greedyfamilies} yield the following  upper bounds.

\begin{corollary} \label{col:gfam}
	Let $H$ be any subgraph of the complement of a graph $G$.
	\begin{itemize}
		\item[(i)] If $G$ is a tree then $ m(G,H,k) \leq \frac{|E(H)|}{k-1}$.	
		\item[(ii)] If $G$ is a simple series-parallel graph with $\chi(G)=3$ then $ m(G,H,k) \leq \frac{|E(H)|}{k-2}$.
		\item[(iii)] If $G$ is a simple bipartite outerplanar graph, then $ m(G,H,k) \leq \frac{|E(H)|}{k-2}$ for $k>2$.		
		\item[(iv)] If $G$ is an $\ell$-tree with chromatic number $\ell+1$ then $ m(G,H,k)\leq\frac{|E(H)|}{k-\ell}$.
	\end{itemize}
\end{corollary}

\section{Robust coloring on paths}\label{subsec:paths}

The seemingly simple setting of paths reflects the complexity of dealing with robust colorings.
L\'opez-Bracho et al.~\cite{BRM05} posed the following conjecture.


\begin{conjecture}~{\rm \cite{BRM05}}\label{conj:paths}
Let $G$ and $H$ be two paths with $n$ edges on the same
vertex set. There exists a $3$-coloring of $G$ such that the
number of monochromatic edges of $H$
is at most $\left\lfloor\frac{n+1}{4}\right\rfloor$.
\end{conjecture}
 The authors pointed out that their upper bound would be tight by the construction in Figure~\ref{fig:twopaths}(a). Theorem \ref{th:two_paths} below proves in the affirmative Conjecture \ref{conj:paths}. We use the number of edges in $G$ that are bridges\footnote{Recall that a \emph{bridge} is an edge whose removal disconnects the graph.} of $G\cup H$, denoted by $s$, to extend the result from $H$ being a path ($s=0$) to $H$ being a vertex-disjoint union of paths ($s>0$); see Figure \ref{fig:twopaths}. Our proof illustrates how having the control on the number of monochromatic edges induced in $H$ by a coloring of $G$ is a difficult problem even for paths.

\begin{figure}[ht]
\begin{center}
 \includegraphics[width=0.6\textwidth]{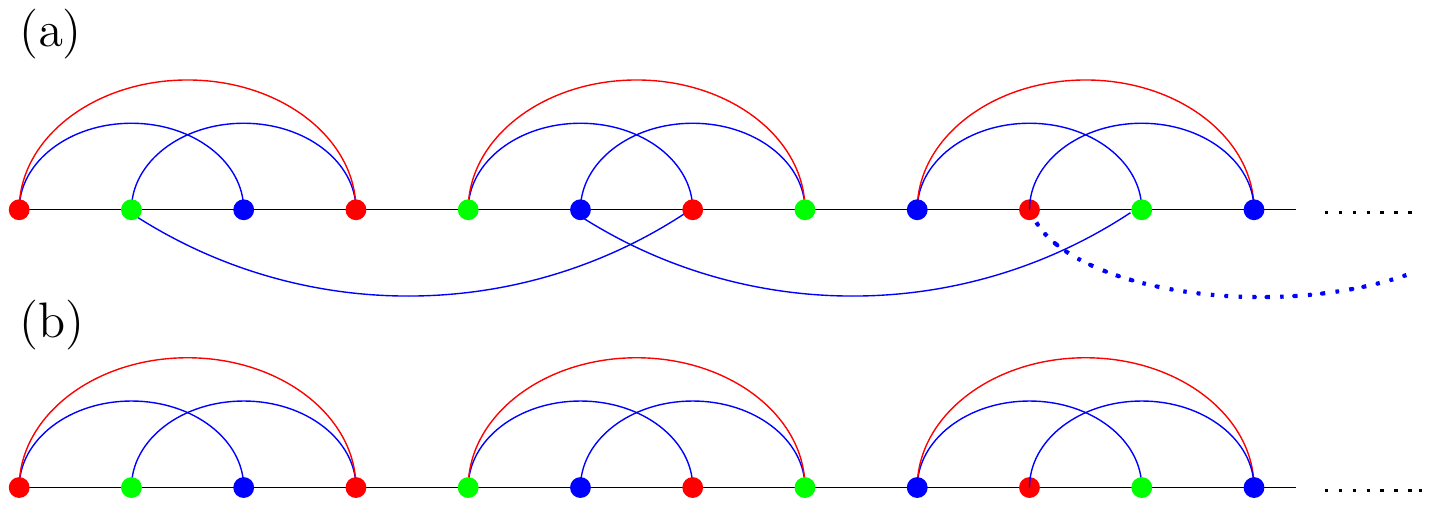}
\end{center}
\caption{(a) The construction given in \cite{BRM05}: $G$ is the horizontal path, and $H$ is a path on the same vertex set (monochromatic edges in red and  bichromatic ones in blue), (b) $G$ is the same horizontal path but now $H$ is a vertex-disjoint union of paths.}
\label{fig:twopaths}
\end{figure}

\begin{theorem}\label{th:two_paths}
	Let $G$ be a path on $n\geq 3$ vertices, and let $H\subseteq\overline{G}$ be a vertex-disjoint union of paths. Then, $$m(G,H,3)\leq \left\lfloor\frac{|E(H)|+1+s}{4}\right\rfloor$$
where $s\geq 0$ is the number of edges in $G$ that are bridges of $G\cup H$.
Moreover, the bound is tight.
\end{theorem}

\begin{proof}
The case $n=3$ is trivial as $m(G,H,3)=0$. It is also easy to check the bound for $n=4$ since $m(G,H,3)=0$ if $G\cup H \subset K_4$, and $m(G,H,3)=1$ if $G\cup H \simeq K_4$. (Note that we could have set $n\geq 2$ but, for $n=2$, the graph $H$ would be empty.) We may assume then that $n>4$ and  also $|E(H)|\geq 4$ (otherwise the result is straightforward).

We first prove by induction on $n$ that the result holds for $s>0$.
Let $e\in E(G)$ be a bridge of $G\cup H$. The graph $(G\cup H) \setminus e$\footnote{We use the standard notation $\widetilde{G}\setminus e$ for the graph that results from deleting an edge $e$ in a graph $\widetilde{G}$.} has two connected components, say $G_1\cup H_1$ and  $G_2\cup H_2$, such that $|E(H)|=|E(H_1)|+|E(H_2)|$ and $s=s_1+s_2+1$, where $s_i$ denotes the number of edges in $G_i$ that are bridges of $G_i\cup H_i$. Hence, $m(G,H,3)=m(G_1,H_1,3)+m(G_2,H_2,3)$, and by induction we have
	$$m(G,H,3)\leq \left\lfloor\frac{|E(H_1)|+1+s_1}{4}\right\rfloor+\left\lfloor\frac{|E(H_2)|+1+s_2}{4}\right\rfloor\leq\left\lfloor\frac{|E(H_1)|+|E(H_2)|+2+s_1+s_2}{4}\right\rfloor$$
which equals the desired upper bound. It may happen that  $H_1$ (or $H_2$) is an empty graph in which case $m(G,H,3)=m(G_2,H_2,3)$ (analogous for $H_2$).

Suppose now that $s=0$. Let $\{u_1, \ldots, u_n\}$ be the set of vertices of the path $G$ (viewed as an horizontal path) ordered from left to right.
We distinguish three cases.

\vspace{0.15cm}

\noindent\emph{Case 1}: \emph{There is a vertex $u_i$, $i\neq n$, satisfying that there exists a unique edge $u_ju_k$ in $H$ such that $j<i$ and $k>i$} (one endpoint of the edge is to the left and the other to the right of $u_i$). 
We proceed by induction on $n$. Let $G_1$ and $G_2$ be, respectively, the sub-paths of $G$ on vertices $\{u_1, \ldots , u_i\}$ and $\{u_{i+1}, \ldots, u_n\}$, that is, $G=G_1\cup G_2 \cup \{u_iu_{i+1}\}$.  Similarly, we consider the graph $H=H_1\cup H_2 \cup \{u_ju_k\}$, where $H_i$ is the subgraph of $H$ contained in $\overline{G_i}$. Every 3-robust coloring of $(G,H)$ can be modified to make the edge $u_ju_k$ bichromatic: it suffices to maintain the coloring of $G_1$ and change the color of $u_k$ with other color in $G_2$, if needed. Thus, $m(G,H,3)=m(G_1,H_1,3)+m(G_2,H_2,3)$, and by induction the result follows.

\vspace{0.15cm}

\noindent\emph{Case 2: Vertex $u_n$ has degree three in $G\cup H$}. 
Again, we use induction on $n$. Let $e_1,e_2$ be the two edges of $H$ incident with $u_n$. Starting from $u_{n-1}$, from right to left, consider the two first right endpoints $u_k, u_j$, $j\leq k$, of edges in $H$; the corresponding edges are denoted, respectively, by $e_3$ and $e_4$.
Let $G_1$ be the sub-path of $G$ on vertices $\{u_1, \ldots , u_j\}$, and let $H_1\subset \overline{G_1}$ be either $H\setminus \{e_i \, | \, 1\leq i\leq 3\}$ (if $u_k\neq u_j$) or $H\setminus \{e_i \, | \, 1\leq i\leq 2\}$ (if $u_k = u_j$). The difference between a $3$-robust coloring of $(G,H)$ and one of $(G_1,H_1)$ relies on at most one edge more that can be monochromatic. Indeed, once the vertices of $G_1$ have been colored, we color the vertices from $u_n$ to $u_{j+1}$: there is one available color for $u_n$ to make $e_1$ and $e_2$ bichromatic, and the induced coloring on $e_4$ always comes from the coloring in $G_1$. Thus, $e_3$ would be the unique edge that could increase the number of monochromatic edges when $u_k\neq u_j$.
Hence, $m(G,H,3)\leq m(G_1,H_1,3)+1$, and the desired bound is obtained by induction.



\vspace{0.15cm}

\noindent\emph{Case 3: The pair $(G,H)$ satisfies neither case 1 nor case 2}.
We present a vertex coloring procedure in which we first go from left to right assigning colors to the vertices so that as long as possible no monochromatic edge is generated in $G\cup H$. If we can color all the vertices, then $m(G,H,3)=0$; otherwise
a \emph{saturated} vertex $u_i$, $3<i<n$, is found: a vertex is saturated if its degree in $G\cup H$ is four and, when it is first visited, three of its neighbours have already been colored with the three available colors. In this case, vertex $u_i$ is not assigned a color, and we continue visiting vertices without coloring until the first \emph{conditioned} vertex $u_j$ is found: a non-colored vertex (at some stage) $u_j$ is conditioned if it is an endpoint of an edge of $H$ whose other endpoint $u_t$
 has already been colored ($t<i$); the edge $u_tu_j$ is said to be \emph{semi-colored}. Observe that at this stage vertices from $u_1$ to $u_{i-1}$ are colored, and those from $u_i$ to $u_n$ are not (including $u_j$). Note also that vertex $u_j$ must exist as $s=0$ and $u_i\neq u_n$. As an example, in Figure \ref{caso_C_i_iii_}, vertex $u_i$ is saturated and vertices $u_j$ and $u_k$ are conditioned.

We next describe how our procedure obtains a proper coloring of $G$ with the property that the total number of monochromatic edges in $H$ equals the number of saturated vertices found during the process. More concretely, when a saturated vertex is visited, there are four or five edges of $H$ involved of which only one has to be monochromatic.

If vertex $u_j$ has two semi-colored edges $e_1$ and $e_2$,  we first assign a color to $u_j$ to make  both edges bichromatic. Then, vertices from $u_{j-1}$ to $u_i$ are colored (from right to left) to maintain the proper coloring in $G$; this gives one monochromatic edge among the two edges of $H$ incident with $u_i$. Assume now that $u_j$ has a unique semi-colored edge $e_1$, and let $u_k$ be the first (from left to right) conditioned vertex in $\{u_{j+1}, \ldots ,u_n\}$; this vertex must exist since otherwise $e_1$ would be the unique edge with one endpoint to the left and the other to the right of $u_i$ (case 1). If $u_k$ has two semi-colored edges $e_2$ and $e_3$ (see Figure \ref{caso_C_i_iii_}(a)),  we first color $u_k$ and $u_j$ so that edges $e_i$, $1\leq i\leq 3$ are bichromatic. Then, from right to left, vertices $\{u_{k-1}, \ldots u_{j+1}\}$ and $\{u_{j-1}, \ldots u_{i}\}$ can be properly colored.
Finally, suppose that vertex $u_k$ has a unique semi-colored edge $e_2$.

\vspace{0.15cm}

\noindent (i) If $u_j$ has either degree 3 in $G\cup H$ or an incident edge $e_3$ with the other endpoint $u_{\ell}$ between $u_i$ and $u_j$ ($i<\ell<j$),  we first visit, from right to left, vertices $\{u_k, \ldots , u_j\}$ assigning colors to make $e_1$ and $e_2$ bichromatic. Then we color, again from right to left, $\{u_{j-1}, \ldots, u_i\}$ so that  $e_3$ (if it exists) is bichromatic; this is possible since we are coloring from right to left so, when $u_{\ell}$ is visited, there are two available colors. Refer to Figure \ref{caso_C_i_iii_}(b).

\vspace{0.15cm}

\noindent (ii) If $u_j$ has an incident edge $e_3$ with the other endpoint $u_{\ell}$ between $u_j$ and $u_k$ ($j<\ell<k$),
we randomly assign to $u_k$ one of the two colors that make $e_2$ bichromatic, and proceed to color from right to left vertices $\{u_{k-1}, \ldots ,u_j\}$ using only the color of $u_k$ and that of the other endpoint of $e_2$. Our aim is that $e_1$ and $e_3$ are bichromatic, but it may happen that, with our assignment, they can not be both bichromatic while maintaining the proper coloring in $G$; this happens for example giving color 1 to $u_k$ in Figure \ref{caso_C_i_iii_}(c). In this case, we change the color of all vertices in $\{u_{k}, \ldots ,u_{j+1}\}$ that have the same color as $u_k$ by the other color that preserves $e_2$ as bichromatic (in our example, we would change color 1 by color 3). Then, we color vertices from $u_{j-1}$ to $u_i$.

\begin{figure}[ht]
\begin{center}
 \includegraphics[width=0.9\textwidth]{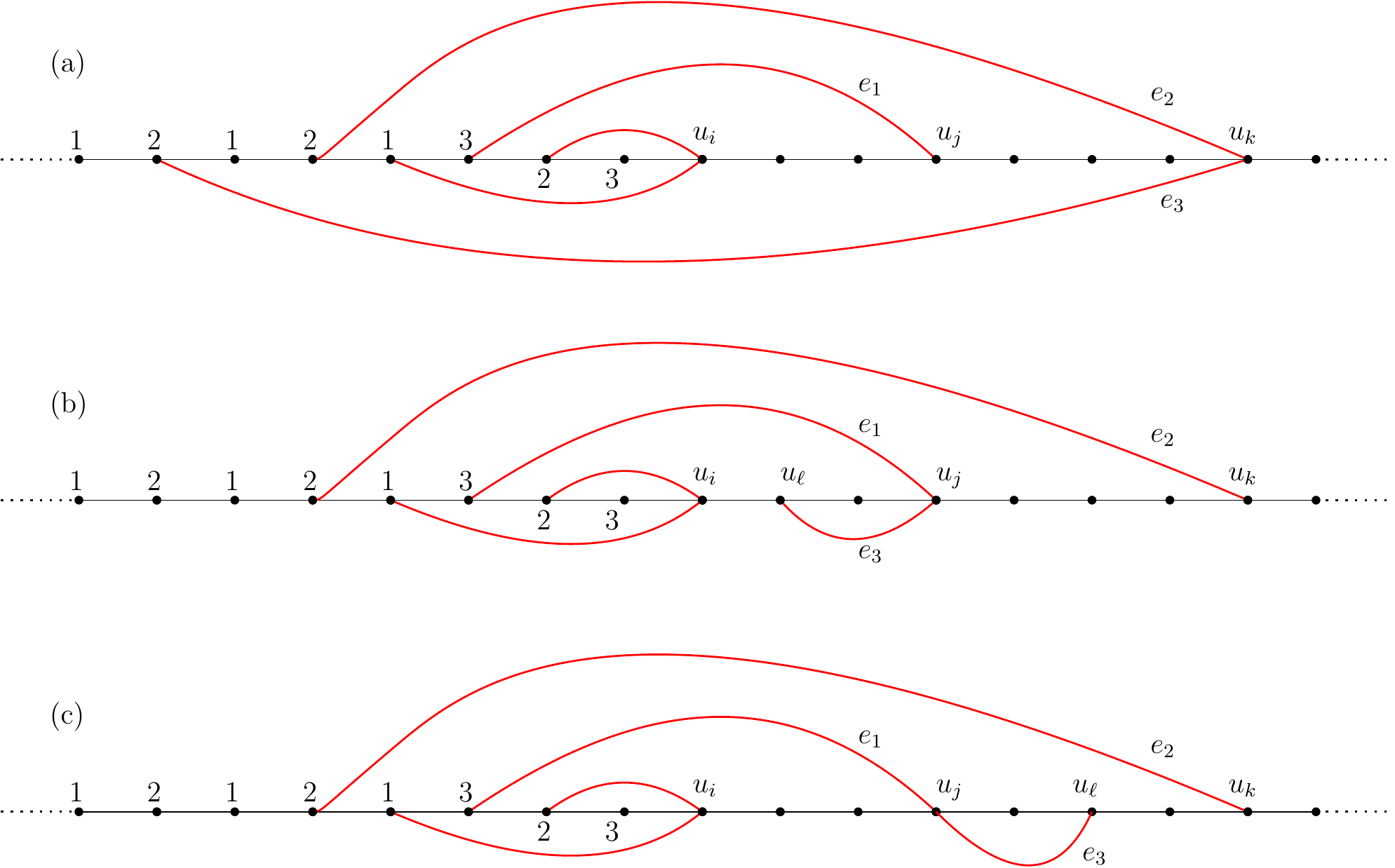}
\end{center}
\caption{Edges of $G$ in black, and edges of $H$ in red; vertices $\{u_1,\ldots, u_{i-1}\}$ have already been colored (colors $1$--$3$). Vertex $u_k$ has two semi-colored edges in (a) and one in (b) and (c); the difference between these two cases is the position of the endpoint $u_{\ell}$ of $e_3$.}\label{caso_C_i_iii_}
\end{figure}

We thus conclude that, when a saturated vertex is visited, our procedure generates one monochromatic edge in $H$, and at least three bichromatic ones. Hence, $m(G,H,3)\leq \left\lfloor\frac{|E(H)|}{4}\right\rfloor$. Figure \ref{fig:twopaths} illustrates examples for $s=0$ and $s>0$ where the bound is tight.
\end{proof}

\section{Concluding remarks}\label{sec:conclusions}

In this work we have presented the first non-heuristic study for general graphs on the robust coloring model. We delved into the connection beetween robust colorings and equitable colorings, encompassing a complexity study. We also obtained the first general bounds on the parameter $m(G,H,k)$, and solved an intriguing conjecture on paths. These are important steps on this difficult and challenging problem that leave different types of open questions for future research:

\begin{itemize}
    \item The problem of deciding whether a general graph has an equitable $k$-coloring with a given number of colors $k\geq 3$ is NP-complete \cite{FJK16}. However, it would be interesting to find a broad class of graphs for which a polynomial time algorithm could be designed. The algorithm could also be applied to robust coloring by means of Theorem \ref{col:rec}.
    \item In order to improve the upper bounds of Section \ref{sec:bounds}, we think that new techniques must be developed, rather than trying to enhance them by using a similar approach to the one presented in this paper.
    \item The proof of Theorem \ref{th:two_paths} shows the complexity of studying the RCP even for paths. Thus, for a better understanding of this coloring model, it would be worth studying if the ideas of that proof could be extended to other families of graphs.
\end{itemize}

\paragraph{Acknowledgments.}
D.G. and A.M. were supported by project PID2019-104129GB-I00/ AEI/ 10.13039/501100011033.


\small
\begin{thebibliography}{99}

\bibitem{ABH14}
C. Archetti, N. Bianchessi, and A. Hertz. A branch-and-price algorithm for the robust graph coloring problem. \emph{Discrete Applied Mathematics} {\bf 165} (2014), 49--59.

\bibitem{aplicacion1} J. P. Barthelemy and A. Guenoche. \emph{Trees and Proximity Representations}. John Wiley Sons, New York, 1991.

\bibitem{BMMPQ07}
E. K. Burke, B. McCollum, A. Meisels, S. Petrovic, and R. Qu. A graph-based hyper-heuristic for educational timetabling problems. \emph{European Journal of Operational Research} {\bf 176}(1) (2007), 177--192.


\bibitem{aplicacion2} G. J. Chaitin. Register allocation and spilling via graph coloring. In: \emph{SIGPLAN'82 Symposium on Compiler Construction}, Boston, Mass. (1982), 98--105.

\bibitem{CZ20}
G. Chartrand and P. Zhang. {\it Chromatic Graph Theory}, 2nd edition, CRC Press, Taylor and Francis Group , Boca Raton, FL, USA, 2020.




\bibitem{DPPP14}
A. Dey, R. Pradhan, A. Pal, and T. Pal.
\newblock The Fuzzy Robust Graph Coloring Problem.
\newblock In: \emph{Proceedings of the 3rd International Conference on Frontiers of Intelligent Computing: Theory and Applications (FICTA) 2014} pp. 805-803, 2014. Part of the Advances in Intelligent Systems and Computing book series (AISC, volume 327).



\bibitem{F06}
H. Furma\'nczyk. Equitable coloring of graph products. \emph{Opuscula Mathematica} {\bf 26}(1) (2006), 31--44.

\bibitem{FJK16}
H. Furma\'nczyk, A. Jastrzebski, and M. Kubale. Equitable colorings of graphs. Recent theoretical results and new practical algorithms.  \emph{Archive of Control Sciences} {\bf 26} (2016), 281--295.




\bibitem{GJ79}
M. R. Garey and D. S. Johnson. \emph{Computers and Intractability: A guide to the theory of NP-Completeness}. W.H. Freeman, New York, USA, 1979.


\bibitem{GJS76}
M. R. Garey, D. S. Johnson, and H. C. So.
\newblock An application of graph coloring to printed circuit testing.
\newblock \emph{IEEE Transactions on Circuits and Systems}, CAS-23:591-599, 1976.

\bibitem{GJS76(1)} M. R. Garey, D. S. Johnson, and L. Stockmeyer. Some simplified NP-complete graph problems, \emph{Theoretical Computer Science}
\textbf{3}(1) (1976), 237–267.


\bibitem{Lim05}
A. Lim and F. Wang.
\newblock Robust graph coloring for uncertain supply chain management.
\newblock In: \emph{Proceedings of the 38th Annual Hawaii International Conference on System Sciences (HICSS'05)} \textbf{3} (2005), pp. 81b.


\bibitem{BRM05} R. L\'opez-Bracho, J. Ramírez, and F. J. Zaragoza-Mart\'{\i}nez. Algorithms for robust graph coloring on paths. In: \emph{Proceedings of the 2nd International Conference on Electrical and Electronics Engineering}, (2005), 9--12.


\bibitem{M73} W. Meyer. Equitable coloring, \emph{American Mathematical Monthly} \textbf{80} (1973), 920-922.



\bibitem{O86}
H. Ogawa.
\newblock Labeled point pattern matching by delaunay triangulation and maximal cliques.
\newblock \emph{Pattern Recognition} \textbf{19}(1) (1986), 35-40.

\bibitem{aplicacioningenieros} P. M. Pardalos, T. Mavridou, and J. Xue. The Graph Coloring Problems: A Bibliographic Survey. \emph{Handbook of Combinatorial Optimization}. Kluwer Academic Publishers  \textbf{2} (1998), 331--395.

\bibitem{SH97}
D. H. Smith and S. Hurley. Bounds for the frequency assignment problem.  \emph{Discrete Mathematics} {\bf 167-168} (1997), 571--582.


\bibitem{YR03}
J. Y\'añez and J. Ramírez. The robust coloring problem. \emph{European Journal of Operational Research} {\bf 148} (2003), 546--558.


\bibitem{YSH17}
B. Yüceo\u{g}lu, G. Sahin, and S. P. M. van Hoesel. A column generation based algorithm for the robust graph
coloring problem, \emph{Discrete Applied Mathematics} {\bf 217} (2017), 340--352.

\bibitem{WX13}
F. Wang and Z. Xu. Metaheuristics for robust graph coloring. \emph{Journal of Heuristics} {\bf 19} (2013), 529--548.



\end{thebibliography}
\end{document}